\documentclass[sigconf, nonacm, screen]{acmart}
\AtBeginDocument{%
  }
\acmConference[HSCC '26]{Hybrid Systems: Computation and Control}{May 11-14, 2026}{Saint Malo, France}
\usepackage{graphicx}
\usepackage{soul}
\usepackage[boxed]{algorithm}
\usepackage[noend]{algpseudocode}
\usepackage{paralist}
\usepackage{epsfig} 
\usepackage{graphicx}
\usepackage[export]{adjustbox}
\usepackage{subcaption}
\usepackage{listings,multicol}
\usepackage{lipsum}
\usepackage{setspace}
\usepackage{comment}
\usepackage{multirow}
\usepackage{tcolorbox}
\usepackage{color}
\usepackage{color, colortbl}
\usepackage{textcomp}
\usepackage{enumitem}
\usepackage{epstopdf}
\usepackage{float}
\usepackage[font={small}]{caption}
\usepackage{wrapfig,lipsum,booktabs}
\usepackage{url}

\newtheorem{definition}{Definition}

\newtheorem{proposition}{Proposition}
\newtheorem{lemma}{Lemma}

\newtheorem{problem}{Problem}
\newcommand{\cl}[1]{\mathcal{#1}}
\newcommand{\bb}[1]{\mathbb{#1}}
\newcommand{\mb}[1]{\mathbf{#1}}
\renewcommand{\qed}{\tiny {\hspace*{\fill}$\Box$}} 

\newcommand{\norm}[2][2]{{\Vert{#2}\Vert}_{#1}}

\newcommand{\bigat}[1]{\Big\vert_{#1}}

\newcommand{\floor}[1]{\lfloor{#1}\rfloor}
\newcommand{\abs}[1]{\lvert{#1}\rvert}
\newcommand{\tupleangle}[1]{\langle{#1}\rangle}
\newcommand{\Piinv}[1][i]{P_{#1}^{-1}}
\newcommand{\Qiinv}[1][i]{Q_{#1}^{-1}}

\newcommand{\Kis}[1][i]{K_{#1}^S}
\newcommand{\alphaides}[1][i]{\alpha_{ ref,#1}}

\newcommand{\R}[1]{\mathcal{R}_{#1}}

\newcommand{\Runsf}{\R{S}\setminus \R{P}}
\newcommand{\xj}[1][j]{\mb{x}_{#1}}

\newcommand{\x}{\mb{x}}
\newcommand{\mbx}[1][x]{\mb{#1}}

\begin{document}
%
\title{Safe Controller Synthesis Using Lyapunov-based Barriers for Linear Hybrid Systems with Simplex Architecture}
%
%
\author{Sunandan Adhikary, Soumyajit Dey}
\affiliation{
\department{Department of Computer Science \& Technology}
\institution{Indian Institute of Technology Kharagpur, India}\country{}}
\email{{mesunandan.kgpian, soumya.cse}@iitkgp.ac.in}

\begin{abstract}
Modern Cyber-Physical Systems often have a two-layered design, where the primary controller is AI-enabled or an analytical controller optimising some specific cost function. If the resulting control action is perceived as unsafe, a secondary safety-focused backup controller is activated. The existing backup controller design schemes do not consider a real-time deadline for the course correction of a potentially unsafe system trajectory, or constrain maximisation of the safe operating region as a synthesis criterion. This essentially implies an {\em eventual safety} guarantee over a {\em small operating region}. 
\par\noindent This paper proposes a novel design method for backup safe controllers (BSCs) that ensure invariance across the largest possible region in the safe state space, along with a guarantee for timely recovery when the system states deviate from their usual behaviour. This is the first work to synthesise safe controllers that ensure maximal safety and timely recovery while aiming at minimal resource usage by switching between BSCs with different execution rates. 
An online safe controller activation policy is also proposed to switch between BSCs (and the primary optimal controller) to optimise processing bandwidth for control computation. To establish the efficacy of the proposed method, we evaluate the safety and recovery time of the proposed safe controllers, as well as the activation policy, in closed loops with linear hybrid dynamical systems under budgeted bandwidth.
\end{abstract}
\maketitle 
\section{Introduction}
\label{secIntro}
Modern autonomous hybrid systems often incorporate learning-enabled components to guide the performance-optimising controllers. Such optimal controllers are typically designed for the underapproximated nominal model of a complex hybrid system and struggle to handle the modelling errors generated from the complex interaction of individual system components, unaccounted wear and tear, etc. Therefore, relying solely on formally verified cost-optimising controllers can reduce robustness against such misapproximated modelling errors, resulting in unsafe and undesired behaviours in their real-world implementation. In this regard, learning-enabled components are often used to enable them to tolerate a certain degree of modelling errors without compromising the desired optimality. However, the inherent opacity of the integrated AI component again poses a challenge to the system safety and its thorough formal verification~\cite{seshia2022toward}. 
Therefore, in both cases, formally verifying safety before real-world implementation is paramount but challenging. 
%
\par
This motivates several research aspects, such as developing model-checking or formal-verification-based methodologies for learning-enabled control components; learning safe control strategies, etc.~\cite {seshia2022toward,hobbs2023runtime}. As nicely summarised in~\cite{seshia2022toward}, developing safe-by-construction or safety-verified AI-based control components for complex system models is hard to achieve. This is why autonomous systems often adopt a {\em simplex architecture}~\cite{sha2001using,bak2014real} (see Fig.~\ref{fig:simplex}). It employs a {\em safety monitor} for monitoring system safety in runtime (bottommost rectangle on the right half of Fig.~\ref{fig:simplex}, labelled with grey text) along with the following two different kinds of controllers. {\em A primary controller}, which may be a learning-enabled controller designed to optimally achieve certain control goals for a given nominal hybrid dynamics, or an analytically designed controller that optimises a desired cost function given a known hybrid dynamics (rectangle labelled with blue text); {\em a backup controller}, which prioritises safety (rectangle labelled in green text). Naturally, the primary optimal controllers (POCs) are largely focused on control performance. Their control actions are passed through a safety monitor, 
which predicts the future system state using a nominal system model. If the prediction is found unsafe, the system switches to the backup safe controller (BSC). This work focuses on a novel BSC design that ensures fast restoration of system safety using limited resources.
\begin{figure}[!b]
\vspace{-8mm}
    \centering
    \includegraphics[width=0.75\columnwidth,clip]{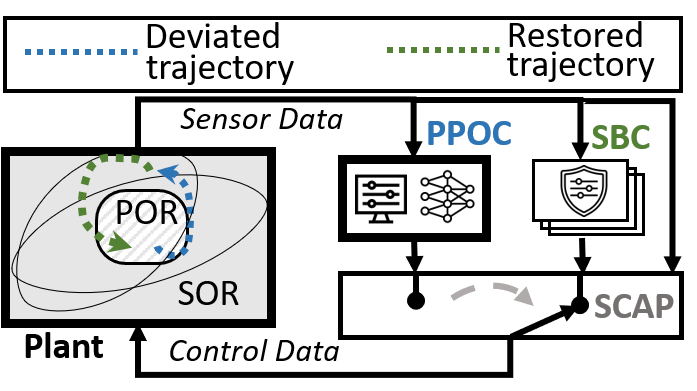}
    \vspace{-0.4cm}
    \caption{Simplex Setup with POC and BSC}
    \vspace{-0.2cm}
    \label{fig:simplex}
\end{figure}
%
%
\par\noindent
Most safe controller synthesis methodologies build upon abstraction-based symbolic assumptions of non-linear closed-loop system dynamics. Consequently, they suffer from the state-space explosion and high approximation error for high-dimensional systems with discrete time constraints. To avoid this, another line of work focuses on \emph{reach-avoid} controller synthesis problems. This involves planning a reference path for a system that avoids certain obstacles and reaches a specific goal region starting from an initial region in its state space~\cite{schurmann2017convex,fan2020fast}. Authors of~\cite{fan2020fast} utilise satisfiability theory to solve this problem. Their methodology requires closed-loop system dynamics along with a primary controller and its corresponding Lyapunov function, which, however, is not easy to find for learning-enabled feedback-controlled systems. 
Also, for the online optimisation problem formulated in~\cite{chen2021fastrack}, its robustness against tracking error may result in infeasibility in runtime.
\par
The aforesaid controller synthesis problems use {\em Lyapunov function (LF)-based constraints} for enforcing Lyapunov stability with desired performance margin and reachability, or constrained optimisation-based techniques to achieve safety. There exist controller synthesis techniques specifically designed to ensure {\em safety by construction} using {\em barrier functions (BF)} to constrain the control operations within a desired safe state-space~\cite{prajna2004safety}.
Combining Lyapunov and BF-based
constraints for computing control actions is therefore naturally explored in works like~\cite{choi2021robust,garg2021robust}
to design safe and stable controllers using quadratic programming (QP).
However, finding such solutions in real-time is computationally challenging. Moreover, finding a suitable BF, which is a prerequisite for such solutions, is an NP-hard problem that requires solving non-convex bilinear matrix inequality constraints~\cite{dai2024verification}.
Works like~\cite{dai2024verification} formulate this barrier synthesis problem as a sum-of-squares (SOS) problem to apply convex SOS-specific relaxations. The authors in~\cite{dai2024verification} also ensure, while synthesising BF, that it is compatible with a candidate LF. This ensures the feasibility of LF-BF-based control synthesis problems in producing a common control solution that adheres to both desired Lyapunov stability and barrier-based safety constraints.
However, unless such controllers enforce {\em timely convergence} as a constraint while synthesising an LF-BF-based control solution, they are not suitable for application in a simplex setup. 
%
In~\cite{garg2021robust}, this is achieved by a QP-based control computation that theoretically ensures a convergence time from its candidate LF. The prerequisite for this solution, though, is a candidate BF. To overcome this issue, authors in~\cite{thomas2018safety} have introduced an iterative controller synthesis technique that uses quadratic LF-based BFs and formulates linear matrix inequalities (LMI) that maximise the LF's $one$-level set (i.e., the set of states in state space for which the quadratic LF evaluates to $1$) with an allowable input range constraint. Such an iterative synthesis algorithm avoids the complexity of SOS or QP-based synthesis by switching between multiple controller-barrier pairs. But this switching might cause instability. To avoid such iterative approaches, authors in~\cite{wang2024convex} propose a single semidefinite program (SDP) to co-design BF and controllers. However, neither of them essentially guarantees safety for the largest possible state space nor promises any desired convergence rate. We intend to propose a controller synthesis solution that bridges these gaps by theoretically ensuring desired convergence from the largest forward-invariant region inside the safe domain (i.e., state trajectories initialised inside this region will always remain confined within it) in a resource-constrained setting.
%
\par\noindent
Usually, in embedded or networked cyber-physical systems (CPSs), tasks are periodically scheduled with fixed priorities, which helps designers to analyse their worst-case response times (WCRTs) and accordingly design the feedback control loops~\cite{bini2008delay}
to adhere to the real-time constraints.
However, new standards for adaptive scheduling and event- or data-driven control support have emerged in recent years to enable efficient utilisation of the available processing and communication bandwidth in resource-constrained compute platforms.
This enables the prioritisation of critical software functions when required, bypassing the interference of predefined scheduling decisions~\cite{adhikary2024revisiting}. 
These are useful in executing bandwidth-intensive controllers using learnable components or optimising complex cost functions in resource-constrained platforms, such as modern CPS architectures, that encourage centralised processing platforms shared among a set of software controllers, rather than federated architectures~\cite{kramer2015real} from earlier times. 
In such platforms, bandwidth efficiency is often managed by adaptively scheduling periodic control tasks by adjusting the periodicity to meet modified run-time performance requirements~\cite{bini2008delay,schinkel2002optimal,adhikary2024revisiting}.
Unlike the event-triggered control invocations, such multi-rate control executions 
still comply with the WCRT-based schedulability analysis to guarantee timeliness.
\par In~\cite{bini2008delay}, an optimal online period assignment method is developed that thoroughly analyses the controller response times and assigns task periodicities within a utilisation budget under a fixed priority scheduler, based on these response time estimates. In a similar vein, the authors in~\cite{sudvarg2025integrated} synthesise safe controllers with suitable periodicities to cope with the computation time incurred by the SOS-based BF computation. This is useful for ensuring timeliness but may lead to unstable switching sequences. To avoid this, we incorporate the idea of LF-guided sampling period switching as used in~\cite{schinkel2002optimal,adhikary2024revisiting} for safe controllers.
\begin{figure}[!b]
    \centering
    \includegraphics[width=0.9\columnwidth,clip]{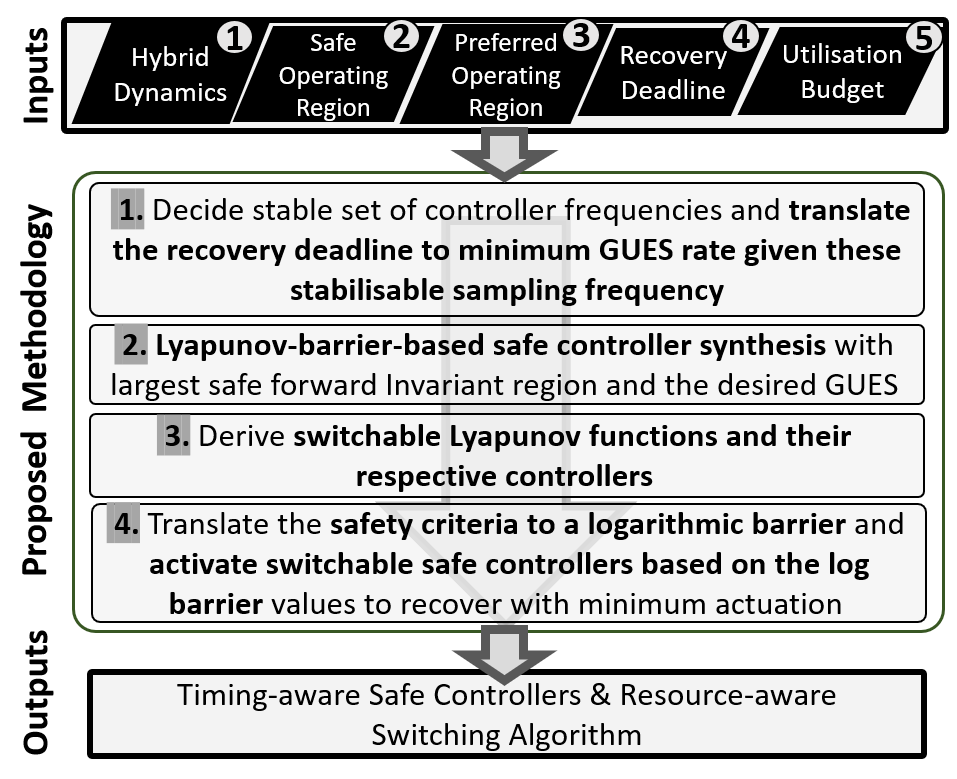}
    \vspace{-0.4cm}
    \caption{Overview of the proposed framework}
    \label{fig:overview}
\end{figure}
\par\noindent 
We consider a hybrid control loop implemented as a simplex architecture in an embedded or networked CPS setting. If the system is perceived to perform unsafely under the operation of a primary optimal controller (POC) as shown with a dashed blue line in Fig.~\ref{fig:simplex}, a backup safe controller (BSC) must be activated for timely recovery to safety (the green dashed line in Fig.~\ref{fig:simplex}). 
We propose a novel framework for synthesis and activation of such safe controllers that enforce {\em (i) safety} and {\em (ii) timely convergence} of the deviated trajectories {\em (iii) with minimal 
communication and computation overheads}. Fig.~\ref{fig:overview} presents an overview of the synthesis framework. It takes as input a hybrid dynamics, linearised around an equilibrium point, along with its safe and preferred operating regions (SOR and POR, respectively), the required recovery deadline, and a processor utilisation upper bound to achieve the following.
 \begin{compactenum}[1.]
\item 
The framework implements a novel {\em safe controller synthesis methodology} (steps 1, 2 in Fig.~\ref{fig:overview}) by formulating a set of linear matrix inequality (LMI) constraints to impose $(i)$ a desired recovery/convergence rate (derived from input 4) of Lyapunov functions (LFs) and $(ii)$ the largest forward-invariant region inside the safe operating region (SOR, input 2) utilising barrier functions (BFs). 
%
\item The methodology is designed to solve the above synthesis problem for selected sampling periodicities while ensuring safe switching among controllers designed with these different periodicities (step 3). This is done while respecting the desired recovery and safety criteria using multiple Lyapunov function (MLF)-based analysis. This helps obtain multiple safe control solutions to recover from a larger subregion within SOR and enables adaptive scheduling decisions for adhering to resource constraints.
%
\item 
Given the set of switchable safe controllers and a processor utilisation budget (input 5), the framework devises an algorithmic framework for {\em state-dependent controller activation policy} (SCAP, in step 4) that intelligently activates these safe controllers based on their forward-invariant regions while adaptively changing their execution rates to adhere to the utilisation budget.
%
\item We finally evaluate our framework for safe controller synthesis and adaptive switching on multiple hybrid control benchmarks.
\end{compactenum}
To the best of our knowledge, this is the first work to map the problem of safe convergence for the timely recovery of hybrid systems in a CPS platform to a resource-sensitive controller switching and adaptive scheduling approach. In the sections that follow, we briefly present the necessary background theory and the proposed methodology built using them.
\section{Prelimineries}
\label{sec:bg}
We can linearise a hybrid control system model $\dot{x}= \mb{f}(x, u)$ near an equilibrium point $x_{ref}$ with linear time-invariant (LTI) dynamics as follows.
\begin{align}
\label{eq:lti_sys}
    &\dot{x}(t) =\Phi x(t) + \Gamma u(t)+w(t),\  y(t) = C x(t)+v(t)\\\nonumber 
    &\hat{x}[k + 1] = A \hat{x} [k] + B u[k] + L(y[k]-C\hat{x}[k]),\ u[k] =-K \hat{x}[k]
\end{align}
Here, the vectors $x \in \mathbb{R}^n$, $\hat{x} \in \mathbb{R}^n$, $y\in \mathbb{R}^m$, and $u\in \mathbb{R}^p$ define the plant state, estimated plant state, output, and control input, respectively. $\hat{x}[k],y[k],u[k]$ are their values at the $k$-th 
($k\in\mathbb{N}$) sampling instance. At the $k$-th sampling instance, the real-time is $t=kh$, where $h$ is the sampling period. 
$w(t)\sim \mathcal{N}(0,Q_w),\ v(t)\sim \mathcal{N}(0,R_v)$ are process and measurement noises. They follow Gaussian white noise distributions with variances $Q_w\in\mathbb{R}^{n\times n}$ and $R_v\in\mathbb{R}^{m\times m}$, respectively. 
The matrices $\Phi= \frac{\partial \mb{f}}{\partial \x}\bigat{\substack{\scriptscriptstyle x=x_{ref}\\\scriptscriptstyle\ u=u_{ref}}}, \Gamma = \frac{\partial \mb{f}}{\partial u}\bigat{\substack{\scriptscriptstyle x=x_{ref}\\\scriptscriptstyle\ u=u_{ref}}}$ are the continuous-time state and input-to-state transition matrices, respectively. The discrete-time counterparts of these matrices are, $A = e^{\Phi h}$, $B = \int\limits^{h}_{0}e^{\Phi t}\Gamma dt$. $C$ is the output transition matrix. The feedback controller module samples the plant output $y[k]$ periodically, once in every $h$ duration, and estimates the plant state $\hat x[k]$ using a Luenberger observer with Kalman gain $L$ to filter the Gaussian noises. The control input $u[k]$ is calculated from the estimated state using $K$ feedback gain. The control input $u[k]$ is communicated via the communication medium that connects the plant and control unit and is applied to the plant through the actuators. 
\par 
We consider that an analytically or learning-enabled controller already exists, that is
designed to optimise specific performance requirements for the given system while operating in the bounded vicinity of an equilibrium point. We express this controller as POC as mentioned earlier, and denote the controller gain as $K$. The bounded region around the equilibrium point, where the POC operates, is termed the hybrid control system's preferred operating region (POR) $\R{P}$ (see Fig.~\ref{fig:simplex}). Since designing a global controller for a nonlinear dynamics is computationally heavy, we consider the controller to be designed for the system dynamics linearised around a certain equilibrium point in the state space.
In Sec.~\ref{sec:safectrl} we explain how a safe controller gain is designed that keeps the system trajectories within the safely operable limits of the system states, which we denote as the safely operable region (SOR) $\R{S}\supset \R{P}$ (see Fig.~\ref{fig:simplex}) for such a linearised hybrid control system. 
\subsection{Convergence in Lyapunov Sense}
\label{subsecLFstab}
A feedback control law is designed with specific stability and safety criteria in mind for the discretised linear difference equation, considering $(A, B)$ is controllable when in $\R{P}$. To quantify how closely the system fulfils the performance criteria while ensuring stability, a control design metric is used that also represents the control objective during controller design. One such design metric that we often use is \textit{settling time}, i.e., the time it takes for the system to maintain a steady output within a fixed error margin around the desired reference value (e.g., within $2\%$ the error band). Hence, the controller has to be designed in such a way that the given settling time requirement is always met. In this work, we consider exponential stability as a performance-driven metric for stability. Given a settling time requirement, we can express it in terms of a desired minimum exponential decay within the mentioned settling time duration so that the system output stays within a $2\%$ error bound of a desired reference. Theoretically, we express it using the notion of the Global Uniform Exponential Stability (GUES) criterion as defined below.
\begin{definition}[{\bf Globally Uniformly Exponentially Stable}]
\label{def:gues}
The equilibrium $x_{ref} = 0$ of the system in Eq.~\ref{eq:lti_sys} is globally uniformly exponentially stable (\textbf{GUES}) if, for 
initial conditions $x[k_0]$, there exist constants $M > 0, \ 0 < \delta < 1$ such that the solution of the system satisfies $\norm{x[k]} \leq M \delta^{(k-k_0)}~\norm{x[k_0]}\leq M e^{\gamma \times k}~\norm{x[k_0]} (\gamma < 0), \ \forall k \geq k_0$ where $\norm{.}$ is the vector norm. \hfill$\Box$
\end{definition}
\par In the {\em Lyapunov sense}, by proving the existence of a continuously differentiable (w.r.t. $k$, $x$) \emph{Lyapunov function (LF)} $V:\bb{R}^n\mapsto \bb{R}$ with the following properties as given below, one can prove the desired stabilisability of the system. 
\begin{align}
\label{eqLF}
    &\text{(\emph{i} ) } V(x_{ref}) = 0,\quad 
    \text{(\emph{ii}) } x\in \R{P}\setminus \{x_{ref}\} \iff V(x)>0,
    \\\nonumber
    &\text{ (\emph{iii}) } x\in \R{P}\setminus \{x_{ref}\} \iff \Delta V= V(\mb{x}[k+1])-V(\mb{x}[k])\leq 0
\end{align}
 Properties (\emph{i}) and (\emph{ii}) ensure that the LF is zero-valued at the reference point $x_{ref}$ and radially diverges in the state space, i.e., $\Vert x\Vert\rightarrow \infty \iff V(x)\rightarrow \infty$. Property (\emph{iii}) 
 ensures the convergence of the LF in the operating state space ${\R{P}}$ along the system trajectory, i.e., $t\rightarrow \infty \iff V(\mb{x}(t))\rightarrow 0$~\cite{astrom97}. This, in turn, confirms {\em local asymptotic stability} (LAS) of the feedback-controlled system when $\mb{x}\in{\R{P}} \setminus \{x_{ref}\}$, i.e., $t\rightarrow \infty \iff V (\mb{x})\rightarrow 0$~\cite{astrom97}. An optimal controller design involves designing a controller with minimum cost, where cost is a function of the states and the inputs. From the perspective of Lyapunov analysis, this implies the existence of an LF $V$ given a desired decay matrix $Q$, i.e.$\Delta V \leq -Q$ where $Q$ is derived from the cost matrices.
 The kinetic energy function for most real-world mechanical and electrical hybrid systems is quadratic and exhibits properties similar to those of the Lyapunov functions as defined in Eq.~\ref{eqLF}~\cite{astrom97}. This makes the kinetic energy functions proper quadratic Lyapunov function (QLF) candidates for these systems, e.g., $V =(\mb{x}- x_{ref})^T P  (\mb{x}-x_{ref})$ with a symmetric positive definite (SPD) matrix $P$ (i.e., $P=P^T>0$). By shifting the origin to $x_{ref}$, we represent the QLF as $V_i=x^TP_ix$, where $x=\mb{x}-x_{ref}$. 
%
%
\par In this work, as performance criteria, we consider the exponential rate of convergence towards a desired reference, in other words, the {\em local exponential stability (LES)}. The state evolution of a system having LES with a desired {\em GUES decay rate} of $\gamma\ (\in(0,\infty))$ can be expressed as defined in Def.~\ref{def:gues}
From the existing literature~\cite{damak2015exponential,adhikary2024revisiting}, in the Lyapunov sense, the LES is defined as below.
\begin{definition}
    \label{defLES}
    The discrete-time LTI system given in Eq.~\ref{eq:lti_sys} is \emph{ exponentially stable} in the Lyapunov sense with an average exponential decay rate $\gamma$ around an equilibrium/reference state $x_{ref}\in\R{S}$ if there exists a QLF candidate $V =(\mb{x}-x_{ref})^T P (\mb{x}-x_{ref})$ (or by shifting the origin to $x_{ref}$, $V_i=x^TP_ix$) that satisfies properties $(i)$ and $(ii)$ from Eq.~\ref{eqLF} along with the following:
    \begin{align}
    \label{eq:LFLES}
        \Delta V  = V (\mb{x}[k+1])-V (\mb{x}[k]) \le - \alpha V_i(\mb{x}[k])\\\nonumber
        \Rightarrow
        (A-BK)^T P  (A-BK) \le (1-\alpha)P 
    \end{align}
    where the decay rate of the QLF $\alpha\in (0,1)$ is computed as $\alpha = 1- e^{-2 \gamma}$.
\end{definition}
\noindent To ensure timely recovery, we compute a desired GUES decay rate $\gamma$ and the desired QLF decay $\alpha\in (1- e^{-2 \gamma},1)$ corresponding to it to formulate the inequality in Eq.~\ref{eq:LFLES}. 
In {\em steps 1, 2} of our methodology (see Fig.~\ref{fig:overview}), we discuss how the recovery deadline is imposed using this QLF decay rate and inequality.
\subsection{Safety Invariance using Barrier Function}
\label{subsecsafety}
\par As for the notion of \emph{safety}, we focus on the {\em forward invariance} criteria among the most prevalent safety conditions considered during controller design. {\em Forward invariance} demands that the states of a hybrid system, when initialised from a {\em compact and convex} space, always remain confined within it in the future. Mathematically, if $\forall k,k^\prime\in\bb{Z}^{+}$ such that, $k^\prime > k,\ \mb{x}[k]\in {\R{S}}\Rightarrow \mb{x}[k^\prime]\in {\R{S}}$ then ${\R{S}}$ is a {\em forward-invariant} region for the autonomous hybrid system. Naturally, it is computationally challenging to ensure forward invariance for a hybrid system across its entire domain when designing a controller.
\par In this work, we use a safety barrier function (SBF) $B:\bb{R}^n\mapsto \bb{R}$ as an indicator of safe and unsafe regions as defined in~\cite{prajna2004safety}. SBF has the following properties: 
\begin{align}
\label{eq:BF}\nonumber
    &\text{(\emph{i} ) } B(x) \le 0, \forall x\in \R{S}\quad 
    \text{(\emph{ii}) }  B(x) = 0, \forall x\in \partial\R{S}
    \\
    &\text{ (\emph{iii}) } \dot{B}(x) \le 0 \forall x\in \R{S}\quad \text{(\emph{iv}) } B(x) > 0, \forall x\notin \R{S}
\end{align}
Here $\partial$ denotes the set of boundary points of the region that follows. Properties (\emph{i}) and (\emph{ii}) ensure that a non-positive value of $B$ denotes the safe region, and property (\emph{iv}) ensures that a negative value of $B$ denotes the unsafe region. Moreover, property (\emph{iv}) ensures the forward invariance of a hybrid system's state trajectory initiated within the zero-level set of $B$. Ensuring the decay of the SBF $B$ to be $\alpha_B\in (0,1)$, i.e., $\dot{B}(x) \le -\alpha_B B$. In steps 2 and 4 of our methodology (see Fig.~\ref{fig:overview}), we use these constraints to synthesise a safe controller and a safe controller switching policy.
\subsection{Hybrid Systems Stably Operating with Multiple Sampling Periods}
\label{sec:multirate_ctrl}
Notice that the discrete-time system characteristic matrices for a system depend on sampling period ($h$), i.e.,  
$\scriptstyle A_h = e^{\Phi h},\ B_h = \int\limits^{h}_{0}e^{\Phi t}\Gamma dt$.
Accordingly, the LQR and Kalman gains, i.e. $K_h$ and $L_h$, respectively,  also change. 
%
%
For multiple sampling periods, $h_i, h_j$, the closed loops can be considered as separate subsystems denoted as $\cl{K}_{h_i}$ and $\cl{K}_{h_j}$ respectively, that use the control gains $K_i$ and $K_j$, respectively. Note that when a stable system changes its sampling period with a new choice of stabilising LQR and Kalman gains designed for the new sampling period, unstable transient behaviour may  exist~\cite{schinkel2002optimal}. For this, we need to constrain the switching to certain sampling period choices such that the overall switched system satisfies the desired decay rate that we discuss in {\em step 3} of our methodology~\cite{liberzon} (as shown in Fig.~\ref{fig:overview}).
%
In the following sections, we present an overview of the proposed solution, followed by a step-by-step description of its components. 

\section{Problem Formulation \& Solution Framework}\label{subsec:meth_ovrvw}
We consider a simplex architecture with a {\em primary optimal controller (POC)} $K$ for the situation when the system states are inside a convex POR $\R{P}$ (see Fig.~\ref{fig:simplex}). 
As discussed earlier, this can be an optimising controller designed for a linearised hybrid system dynamics, or a learning-enabled controller that works effectively with certain objectives for a nominal system dynamics. 
It safely operates when initialised within the POR $\R{P}$ (see the obliquely striped region in Fig.~\ref{fig:simplex}). 
When the system state trajectory exits $\R{P}$, a backup safe controller (BSC) gain $K^S$ must be activated before the trajectory becomes unsafe or evolves beyond the SOR $\R{S}$ (the blue dashed trajectory in Fig.~\ref{fig:simplex}). 
In this simplex setting, we intend to synthesise a set of such BSCs $\mb{K}^S=\{K^S_i, \forall h_i\}$ for different sampling periods $h_i$ (for $ i\in\mathbb {Z}^{>0}$) that always ensure the following criteria.
     {\bf (i)} 
     The system must recover within a desired recovery deadline $\delta t$, i.e., $\mbx(\delta t)\in \R{P} \forall \mbx(0)\in \Runsf$, and
     %
     {\bf (ii)} 
     while recovering, the trajectory must not leave $\R{S}$, i.e., $\mbx(t)\in\R{S} \forall t\in [0,\delta t]$ (see the green dashed trajectory in Fig.~\ref{fig:simplex}).
We also intend to devise an algorithmic framework for a safe {\em state-dependent controller activation policy (SCAP)} to switch between the BSCs while maintaining the aforementioned criteria
enforcing the resource limit. 
%
\par\noindent
The overall solution methodology to solve the aforementioned problems is outlined in Fig.~\ref{fig:overview}. Following are the inputs to the proposed solution framework: {\bf 1.} the hybrid system dynamics, {\bf 2.} the safety boundary $\R{S}$, {\bf 3.} the {operating region boundary} for the POC $\R{P}$, {\bf 4.} the recovery deadline $\delta t$ and {\bf 5.} a processor utilisation upper bound $U_b$.
Given these inputs, the methodology outputs a set of BSCs that satisfy criteria {\bf (i)} and {\bf (ii)}, as mentioned in the last paragraph, and a SCAP algorithm for monitoring safety and switching between the BSCs in a resource-sensitive manner. 
In the following sections, we explain the primary components of the proposed framework.
\subsection{Step 1: Minimum Exponential Decay Rate Derivation from Recovery Deadline}
\label{sec:alpharef}
Consider that we intend to synthesise a controller gain that recovers states, starting from anywhere in $\Runsf$ to the POR $\R{P}$ within $\delta t$ time. This recovery deadline can be converted to the number of required sampling instances, $\delta k$, for a constant sampling period, $h$. As the sampling period changes, $\delta k$ also changes since, for a sampling period $h_i$, the recovery deadline $\delta k_i = \floor{\frac{\delta t}{h_i}}$. To make the formulation easier, we translate the requirement of the recovery deadline $\delta k$ into the QLF decay $\alpha_i$ corresponding to the desired recovery specification. 
\begin{proposition}
\label{propqlfdecay}
    If there exists a quadratic Lyapunov function with decay rate $\alpha_i\in (\alpha_{ref,i},1)$ such that a trajectory, starting from anywhere in $\Runsf$ recovers back to $\R{P}$($\subset \R{S}$) within $\delta k$ sampling iterations;
    then the minimum required decay rate $\alpha_{ref,i}\in (0,1)$ is:
    \begin{align}
        \label{eq:Alpha}
       \alpha_{ref,i} = 1-\Bigg(\frac{\min\limits_{\xj \in \partial \R{P}}{\norm{\xj}}}{\max\limits_{\xj \in \partial \R{S}}{\norm{\xj}}}\Bigg)^{\frac{2}{\delta k}}
    \end{align}  
\end{proposition}
\begin{proof}
We know from Def.~\ref{def:gues}, for a hybrid system with GUES decay rate of $\gamma>0$, its states evolve as $\norm{\mb{x}[k+\delta k_i]} \leq M\norm{\mb{x}[k]} e^{-\gamma \times \delta k_i}$ $\Rightarrow$
$\frac{\norm{\mb{x}[k+\delta k_i]}}{\norm{\mb{x}[k]}}< e^{-\gamma \times \delta k_i}$. By replacing this in the equation of $\alpha$ in Def.~\ref{defLES} we bound the desired QLF decay rate as:  
$\alpha_{i} \geq  1-\Big(\frac{\norm{x[k+\delta k_i]}}{\norm{x[k]}}\Big)^{\frac{2}{\delta k_i}}$. Since $\alpha_{ref,i}$ is the minimum required decay that ensures recovery from anywhere in $\Runsf$ to $\R{P}$, it must consider the worst case, where the most unsafe state recovers to the farthest safe state.
To account for the worst case, we replace the initial state $x[k]$ with the outermost unsafe point or the farthest point in the boundary of SOR  $\partial \R{S}$ (or in the outer boundary of $\Runsf$), i.e., $\max\limits_{x\in \R{S}} x$. Similarly, the final state   $x[k+\delta k_i]$ is replaced with the innermost point in the POR boundary (or in the inner boundary of $\Runsf$), i.e. $\min\limits_{x\in \R{P}} x$. Replacing these in the last inequality gives us Eq.~\ref{eq:Alpha}. Since $0\le \norm{x\in \Runsf}\le \max\limits_{\xj \in \partial \R{S}}{\norm{\xj}}$ and $\frac{2}{\delta k}>0$; $\Bigg(\frac{\min\limits_{\xj \in \partial \R{P}}{\norm{\xj}}}{\max\limits_{\xj \in \partial \R{S}}{\norm{\xj}}}\Bigg)\le \Bigg(\frac{\min\limits_{\xj \in \partial \R{P}}{\norm{\xj}}}{\norm{\x \in \Runsf}}\Bigg)\Rightarrow \alpha_{ref,i}\le \alpha_{i}$. Naturally, the existence of a QLF with a minimum decay rate of $\alpha_{ref,i}$ can recover system states from $\Runsf$ to $\R{P}$, 
\end{proof} 
Prop.~\ref{propqlfdecay} gives us the desired minimum QLF decay rate for a particular sampling period. As discussed in Sec.~\ref{sec:multirate_ctrl}, we choose a set of linearised hybrid systems $\mb{\cl{K}}$ with different sampling periods such that characteristic matrices for each system $\cl{K}_{h_i}\in\mb{\cl{K}}$ discretised with $h_i$ sampling period satisfy Eq.~\ref{eq:LFLES}, when the QLF decay rate $\alpha$ is replaced with $\alpha_{ref, i}$. In the following section, we derive the constraints to impose the desired convergence rate and safety while synthesising controllers for a particular sampling period. 
\subsection{Step 2: Constraint Synthesis for Desired Recovery Rate and Safety Guarantee} \label{sec:safectrl}
$\blacksquare$ {\bf Recovery:} We intend to synthesise a safe controller gain $K_i^S$ with $h_i$ sampling period,  
to ensure recovery within $\delta k_i$ sampling iterations. As mentioned above, solving the constraint in Eq.~\ref{eq:LFLES} with parameters specific to a certain sampling period $h_i$ is sufficient for this purpose. Here, the parameters for the sampling period $h_i$ are the desired decay rate $\alpha_{ref, i}$ for the {\em quadratic Lyapunov function (QLF)} $V_i=x^T P_i x$ and the discrete state transition matrix of the closed-loop $(A_i-B_iK_i^S)$. This makes the problem harder to solve than a quadratic program, since the equation contains three unknowns: ${K_i^S}^T$, ${K_i^S}$, and $P_i$ in a single term. 
As mentioned in Sec.~\ref{secIntro}, we intend to synthesise a semidefinite constraint using the substitution and {\em Schur complement} methods from this constraint, to ensure a feasible and lightweight solution. This gives us the following lemma.
\begin{lemma}[timely recovery guarantee]
\label{lemLES}
Consider a hybrid system linearised around an equilibrium point $x_{ref}$ as presented in Eq.~\ref{eq:lti_sys} having constant state and input-to-state transition matrices, $A_i$ and $B_i$, respectively, for all states in a convex and closed safe operating region $\R{S}$ 
when discretised with a sampling periodicity $h_i$. 
There exists a feedback controller gain $K_i^S$ such that the system states initialised from $\Runsf$ are guaranteed to return to its convex preferred operating region (POR) $\R{P}$ within a $\delta k_i$ number of sampling periods under the operation of $K_i^S$ only if the following LMI in Eq.~\ref{eq:LMIqlf} is satisfied, where $P_i$ is a symmetric positive definite (SPD) matrix and the bound on $\alpha_i\in [\alphaides,1)$, i.e., $\alphaides$ satisfies Eq.~\ref{eq:Alpha} for a positive $\epsilon <<1$.
%
\begin{align}
\label{eq:LMIqlf}
    \begin{bmatrix}\Piinv & A_i \Piinv -B_i Z_i\\
        (A_i \Piinv -B_i Z_i)^T & (1-\alpha_i)\Piinv\end{bmatrix} \ge \epsilon, Z_i = \Kis \Piinv
\end{align}
\end{lemma}  
\begin{proof} 
Using the parameters specific to the sampling period $h_i$ in Eq.~\ref{eq:LFLES} we get : $(A_i-B_iK_i^S)^T P_i (A_i-B_iK_i^S) \leq (1-\alpha_i)P_i$ and $P_i > 0$ . 
%
Since $\alpha_i \in (\alphaides,1)$ and $\alphaides >0$, the second inequality $P_i > 0 \Rightarrow (1-\alpha_i)P_i > 0$. 
Now that $P$ and $(1-\alpha)P$ both are SPD, they are self-adjoint matrices (i.e., $P=adjoint(P)$). Therefore, we can apply the Schur complement method~\cite{boyd2004convex} on the combined inequality, which gives us the following LMI.  
$$\begin{bmatrix}P_i^{-1} & (A_j P_i+B_jK_i^S)\\
        (A_j P_i+B_j K_i^S)^T & (1-\alpha_i)P_i \end{bmatrix} \ge \epsilon > 0$$ 
Here, a small positive $\epsilon<<1$ is used to introduce equality. Multiplying the LHS of above LMI with $\begin{bmatrix}I & 0\\
0 & P_i^{-1} \end{bmatrix}>0$ from both sides (pre and post multiplications), and substituting $K_i^S\Piinv$ with $Z_i$, we get LMI~\ref{eq:LMIqlf}.
Therefore, solving LMI~\ref{eq:LMIqlf} for a desired $\alpha_i$ is a sufficient condition for the existence of the QLF $V_i = (x-x_{ref})^T P_i (x-x_{ref})$. This proves that using the controller gain $\Kis$ that is a solution to LMI~\ref{eq:LMIqlf} for any $\alpha_i\in(\alphaides,1]$ ensures that under the operation of $K_i^S$, system states initialised from $\Runsf$ will return to $\R{P}$ within a $\delta k_i$ number of sampling periods as stated in Prop.~\ref{propqlfdecay}.
\end{proof}
Note that Eq.~\ref{eq:LMIqlf} is a semidefinite constraint with $\Piinv$ and $Z_i$ as unknown variables. 
We can solve Eq.~\ref{eq:LMIqlf} iteratively for different $\alpha_i\in(\alphaides,1]$ until we get a satisfying solution $\Kis$, where  $\alphaides$ is computed using Prop.~\ref{propqlfdecay}.
This kind of stability that hinges on the existence of a QLF $V_i=(\mb{x}-x_{ref})^T P_i(\mb{x}-x_{ref})$ for all $\mb{x}\in \Runsf$ is also known as {\em quadratic stability} of the system around the equilibrium/desired reference point $x_{ref}\in \R{S}$~\cite{boyd2004convex}. However, this {\em does not ensure that the controller will keep the system states safe  (i.e., always inside SOR)} while recovering within the desired deadline. Therefore, we need to ensure that the level sets of $V_i$ corresponding to $\Kis$ remain inside $\R{S}$ while maximising them to ensure safe operation across the largest region in $\R{S}$.
\par\noindent $\blacksquare$ {\bf Safety:} This brings us to the second LMI formulation that ensures safety under the operation of $\Kis$. For this purpose, we use the notion of  {\em forward invariance}, i.e., the system property that always holds as the time progresses. The safety barrier function, or SBF, has a forward-invariant zero-level set, i.e., the trajectories initialised inside the zero-level set remain inside in the future. Since we are using QLFs, their level sets (explained in Sec.~\ref{sec:bg}) are always ellipsoidal. A $c$-level set of QLF $V_i$ is ${\cl{E}_c}: \mb{x}^T \Qiinv \mb{x} \leq c$ and it basically represents an ellipsoid with its centre at $x_{ref}=0$ and radius $c Q_i$. If this is the largest level set of the QLF corresponding to a controller gain, inside the safe state space $\R{S}$, then the zero-level set of $V_i-c$ works as a safety barrier. Moreover, the level sets of any LF are forward-invariant. Therefore, $V_i-c$ becomes a natural SBF candidate for the safe controller $\Kis$. Hence, we define the safety in the Lyapunov sense as below.
\begin{definition}
\label{defFI}
If there exists a candidate QLF  $V = x^T Q^{-1} x$ ($Q> 0$ and symmetric) for a linearised time-invariant hybrid system given in Eq.~\ref{eq:lti_sys} operating under feedback gain $K$ and the largest ellipsoidal $c$-level set of $V$ inside the safe operating region $\R{S}$, is ${\cl{E}_c}: V(x) = x^T Q^{-1} x \leq c$, then we signify this region ${\cl{E}_c}$ as the \emph{largest} safe forward-invariant region (SFIR) for controller $K$ as the system remains inside when initialised from anywhere inside ${\cl{E}_c}$.
\end{definition}
\noindent Naturally, our goal while synthesising the controller will be to maximise the largest SFIR for a required decay rate and $\Kis$ and confine it inside $\R{S}$ to formulate a safety barrier. The idea is depicted in Fig.~\ref{fig:lyapunov}, where there are two such ellipsoids ${\cl{E}_{c_1}}\subset \R{S}$ and ${\cl{E}_{c_2}}\subset \R{S}$, which are the largest SFIRs corresponding to two QLFs $V_1$ and $V_2$ respectively, which have decay rates $\alpha_1,\alpha_2$. This ensures all the states originating inside the ellipsoidal approximations of $\R{S}$ are bounded within $\R{S}$. We propose the following lemma to enforce this while synthesising an BSC gain $K_i^S$ having QLF with the desired decay rate. 
\begin{lemma}[safety guarantee]
\label{lemSafety}
Consider a hybrid system linearised around an equilibrium point $d$, discretised with a sampling periodicity $h_i$, having constant state and input-to-state transition matrices, $A_i$ and $B_i$, respectively, for all states in its convex and closed safe operating region (SOR) $\R{S}$ (presented in Eq.~\ref{eq:lti_sys}). $A_S (x-d)\le b_S$ is the polytopic representation of $\R{S}$.
The QLF\; $V_i=(\mb{x}-d)^T \Qiinv (\mb{x}-d)$ under the operation of a feedback controller gain $\Kis$ has the largest safe forward-invariant $c_i$-level set $\cl{E}_{c_i}$ inside $\R{S}$ when $\overline{Q_i}$ is maximised under the following LMI constraints. 
%
\begin{align}
\label{eq:LMIsafe}
&\begin{bmatrix}\overline{Q_i} & \overline{U_i}\\
    \overline{U_i}^T & I\end{bmatrix} \leq 0 \text{ for a symmetric $\overline{U_i}=c_i U_i > 0$}\\ \label{eq:LMIsafe2}
& and\ \norm{{a_S^{(j)}}^T\overline{U_i}}+ {a_S^{(j)}}^T d \le b_S^{(j)}\ \forall \text{ $j$-th rows of }A_S, b_S
\end{align} 
\end{lemma}  
\begin{proof}
   The polytope $\R{S}$ can be expressed using the following inequalities: ${a_S^{(j)}}^T(\mbx -d)\le b_S^{(j)}$ for each $j$-th row of $A_S,b_S$. 
   For an ellipsoid $\cl{E}:(\mbx-d)^T P_i (\mbx-d)\le c_i\Rightarrow (\mbx-d)^T \overline{P_i} (\mbx-d)\le 1\Rightarrow  
   \norm{\overline{U_i}^{-1} (\mbx-d)}\le 1\ \Rightarrow \norm{z}\le 1$, 
   when $z = \overline{U_i}^{-1}(\mbx - d)$ and $\overline{U_i}^T\overline{U_i}=\overline{Q_i}$. Using the Schur complement method, we translate the equality $\overline{U_i}^T\overline{U_i}=\overline{Q_i}$ to $\begin{bmatrix}\overline{Q_i} & \overline{U}_i\\    \overline{U}_i^T & I\end{bmatrix}=0$.  Therefore $\begin{bmatrix}\overline{Q_i} & \overline{U}_i\\    \overline{U}_i^T & I\end{bmatrix}=0$.
   %
   %
   With affine transformations 
   the polytopic inequalities for each $j$-th row of $A_S,b_S$ in $z$ coordinate become, ${a_S^{(j)}}^T\overline{U_i}z\le b_S^{(j)}$. 
   Since maximum value of $\norm{z}$ is 1 as shown earlier, therefore $\norm{{a_S^{(j)}}^T\overline{U_i}}\le b_S^{(j)}$. Since $a_S^{(j)}, b_S^{(j)}$ are constants, and $\overline{U_i}$ is proportional to $\mbx$, to maximise the area of the ellipsoid inside the polytope $\R{S}$, we need to maximise $\overline{U_i}$ or $\overline{Q_i}$respecting the above inequality, i.e., the LMI in Eq.~\ref{eq:LMIsafe2}. Since $\begin{bmatrix}\overline{Q_i} & \overline{U}_i\\    \overline{U}_i^T & I\end{bmatrix}=0$, while maximising we can consider the LMI in Eq.~\ref{eq:LMIsafe} as the lower bound for $\overline{U_i}$.
   %
\end{proof}
The maximisation of $c_iQ_i$ in the presence of the above LMIs ensures that the safe controller corresponding to the QLF $V_i = x^T Q_i^{-1} x$ ensures that the trajectories initialised within SOR $\R{S}$ remain confined inside it. 
As defined in Def.~\ref{defFI}, here $V_i-c_i$ acts as the SBF $B_i$ for synthesising the BSC gain $K_i^S$ given a linearised hybrid system discretised with $h_i$ sampling period. Synthesising $K_i$ using the constrains and objectives outlined in Lemma~\ref{lemLES} and~\ref{lemSafety}, we can ensure that any potentially unsafe states $x\in \Runsf$ returns back to the POR $\R{P}$ within ${\delta_{k}}_{i}$ sampling iterations without leaving the largest SFIR for $K_i^S$, i.e., $\forall k\in [0, {\delta_{k}}_{i}]\ x[k]\in \R{S}$. The sampling periodicities of these BSCs are decided based on their decay rates and SFIRs, as discussed in the next step.
%
\subsection{Step 3: Synthesis of Multi-Rate BSCs and Safe \& Stable Switching Logic} \label{sec:switchctrl}
We consider a resource-constrained CPS setting, where control tasks are scheduled periodically based on their operational priorities. A limited amount of processing bandwidth is shared among multiple controllers in such setups. Safety and timely recovery are not the only concerns while deploying the solution in such a setting; the solution also needs to be resource-aware. 
\par \noindent $\blacksquare$ \textbf{ Overall multi-rate BSC Synthesis Problem:} 
To enable minimum usage by limiting the processing bandwidth, we synthesise multiple such BSCs for different sampling periodicities. This is done by choosing a base periodicity $h_0$, less than the operating periodicity of POC (usually $1/h_0 \approx $ 20 to 30 times the bandwidth frequency observed from the Bode plot of the open-loop transfer function) and its multiples, i.e., $2h_0, 3h_0,\ldots$ These are chosen such that the closed loops are stable, and BSCs for each of these sampling periods can be synthesised by optimising the objectives in the presence of the constraints derived in Lemma~\ref{lemLES},~\ref{lemSafety}. Therefore, for a hybrid system linearised around the equilibrium point $x_{ref}\in\cl{X}_S$ as defined in Eq.~\ref{eq:lti_sys}, to synthesise an BSC gain $K_i^S=Z_iP_i^{-1}$ for a sampling period $h_i=m h_0$, $m\ge 1$, we need to solve the following constrained optimisation problem: 
\begin{problem}[Overall Synthesis Problem]
\label{prob:synth}\;
\par \noindent
$\square$ Find $P_i$ by solving the following semidefinite problem.
\begin{align}\label{eq:opt}
    &\hspace{1.2cm}\max P_i^{-1} 
    \text{ under the following LMI constraints}:\\\nonumber
    \text{(\ref{eq:opt}.1) }& \text{ Eq.~\ref{eq:LMIqlf} with $\alpha_{ref,i}$ from Eq.~\ref{eq:Alpha},}\\\nonumber
    \text{(\ref{eq:opt}.2) }&  \text{ Eq.~\ref{eq:LMIsafe}, Eq.~\ref{eq:LMIsafe2} with $Q_i=P_i^{-1}$ and for a constant $c_i$}. 
    \end{align}
%
$\square$ Find a $c_i$ for the derived $Q_i$ such that the ellipsoid $\cl{E}_{c_i}:x^T \Qiinv x\le c_i$ is the largest SFIR for BSC gain $K_i$ as defined in Def.~\ref{defFI} or the SBF $B_i = x^T \Qiinv x-c_i$, where $x = \mbx - x_{ref}$.
\qed
\end{problem}
Since the objective and constraints in this problem are convex, linear and involve searching for an SPD matrix, it is a semidefinite programming problem. We solve it with $c_i=1$, i.e., $Q_i=\overline{Q_i}$ in Eq.~\ref{eq:LMIsafe}. 
If the ellipsoid $\cl{E}_{c_i}:x^T \Qiinv x\le c_i$ is inside and smaller than the convex SOR polytope $\R{S}:A_Sx+b_S\le 0$, $c_i$ should be increased until $\cl{E}_{c_i}$ inscribes $\R{S}$. 
If $\cl{E}_{c_i}$ is larger than $\R{S}$ and covers it, $c_i$ should be decreased until $\cl{E}_{c_i}$ inscribes $\R{S}$. 
\par An example of such synthesised BSCs operating in their SFIRs is depicted in Fig.~\ref{fig:lyapunov}. Here we switch between two BSC gains, namely $K^S_2$ and $K^S_1$. In the diagram on the left side of Fig.~\ref{fig:lyapunov} (i.e., the diagram in the middle of Fig.~\ref{fig:overview}), the trajectory is marked with a blue dotted line. At $time =0$, the BSC gain $K^S_2$ is used as the system state $x[0]$ is inside the SFIR of $K^S_2$ $\cl{E}_{c_2}$ (marked with black boundary). Notice the diagram on the right side of Fig.~\ref{fig:lyapunov}; the SFIR $\cl{E}_{c_2}$ inscribes the safe region $\R{S}$ (black outlined ellipse inscribes the grey region). How $\cl{E}_{c_2}$ is computed as the SBF, using the QLF $V_2$, is also presented in the black text box in the middle. The blue star shows when the BSC gain $K^S_1$ is activated as the system trajectory enters the SFIR of $K^S_1$. i.e., $\cl{E}_{c_1}$ (depicted as the ellipse with grey boundary). A black star marks the point where the system reverts to the POC as the trajectory recovers inside the POR. 
\par \noindent $\blacksquare$ \textbf{ Feasibility of the Synthesis Problem in Eq.~\ref{eq:opt}:} 
Note that this constrained optimisation problem might not be feasible for every desired $\alpha_{ref, i}$ and for any periodicity $h_i$. A multi-rate synthesis solution also helps in this case. We can search for a base periodicity $h_0$ or different values of $m$ such that the optimisation problem becomes feasible. 
If the base periodicity does not support the minimum QLF decay rate $\alpha_{ref,0}$, we may need to extend the deadline $\delta t$ (or $\delta k_0$) and synthesise BSCs for a lower $\alpha_{ref,0}$. 
\par \noindent $\blacksquare$ \textbf{ BSC Activation Policy Design:} Given the set of synthesised multi-rate BSCs, a stable and safe {\em state-dependent controller activation policy or SCAP} must be designed that switches between these controllers. It switches to an BSC such that the current state belongs to its SFIR (see Def.~\ref{defFI}). It should also enable the BSCs, while the system trajectory moves beyond $\R{P}$ (but within $\R{S}$) and switch back to POCs as the trajectory recovers inside $\R{P}$. During these state-dependent switching decisions, SCAP should therefore ensure the following: {\bf 1.} while switching from the POC to any of these BSCs, the stability is not hampered; {\bf 2.} while switching between the BSCs with different sampling periodicities, the safety and recovery rate (as derived in Prop.~\ref{propqlfdecay}) is maintained.
\par Note that each of the POCs, denoted using $K_i$, is ensured to be stable inside the POR $\R{P}$, and each of the BSCs is guaranteed to be stable and safe within SOR $\R{S}$, we switch to the BSCs from the POCs as the system states evolve outside $\R{P}$. As shown in~\cite{liberzon} (Chapter 3.4.2, Condition 1), in such a scenario where the subsystems have different QLFs, a switching logic exists that ensures quadratic stability even between unstable subsystems. However, in our case, we have to meet a desired recovery deadline or guarantee a desired exponential convergence rate while switching between BSCs safely. 
%
\begin{lemma}[Safe \& Stable Switchability Guarantee]
\label{lemSwitch}
    Consider a switchable set of BSCs $\mb{K}^S=\{K_1^S, K_2^S,\ldots\}$ synthesised by solving the optimisation problem in Prob.\ref{prob:synth} for different periodicities. Given QLF of $K^S_i$, $V_i$,  its decay rate, $\alpha_i\in [\alpha_{ref,i},1]$, and its SBF $B_i=V_i-c_i$, a state-dependent controller activation policy $\pi(x)$ as defined below, ensures the recovery deadline of $\delta k$ from anywhere in $\R{S}$ (as defined in Def.~\ref{defFI}) while switching between these controller gains. 
    \begin{align}\label{eq:scap}
        \pi(x) =arg \max\limits_{K_i^S\in \overline{\mb{K}^S}} \alpha_iV_i(x),\  \overline{\mb{K}^S}= \{K^S_i\   s. t.\ V_i(x)-c_i\le 0\}
    \end{align}
\end{lemma}
\begin{proof}
    Consider the state-dependent switching function $\pi(x[k])=j$, which denotes that at the $k$-th sampling iteration, $K_j^S \in \mb{K}^S$ is activated. Consider that the switched system starts from the controller $K^S_{\pi(x[0])}$ that has a QLF $V_{\pi(x[0])}(x[0])$. We simplify these notations as $K^S_{\pi_0}$ and $V_{\pi_0}$, respectivley. Now, safety is always ensured since at any $k$-th sampling instance, we always choose from the subset $\overline{\mb{K}^S}= \{K^S_j\ s. t.\ V_{\pi_k}-c_{\pi_k} = B_{\pi_k}\le 0\}$, i.e., the BSCs with a non-positive SBF values ensuring safety 
    as defined in Def.~\ref{defFI} (used similar notation to express the SBF as well).
    \par Following Eq.~\ref{eq:scap}, after multiple state-dependent switching between BSCs, the discrete-time QLF decay after $\delta k$ sampling periods is as follows: 
    $\Delta_{\delta k} V \le \sum\limits_{k\in [0, \delta k]}(-\alpha_{\pi_k}) V_{\pi_k}$. A switchable QLF at $k$-th sampling period can be expressed using a switching constant $\mu_{\pi_{k-1}} > 0$ as $V_{\pi_k} \le \mu_{\pi_{k}}V_{\pi_{k-1}}$. When we do not switch to a different BSC at the $k-1$-th iteration, $\mu_{\pi_k} = 1$. Therefore, considering all BSC swithings in duration $[0,\delta k]$, we get $\Delta_{\delta k} V \le \sum\limits_{k\in [0, \delta k]}(-\alpha_{\pi_k})\mu_k V_{\pi_0}$.
    %
    Our switching policy ensures $-\alpha_{\pi_k} V_{\pi_k}=-\alpha_{\pi_k} \mu_{\pi_k}V_{\pi_{0}}\le -\alpha_{i} V_{i}\le -\alpha_{ref,0} \mu_{i}V_{0}$ $\forall K^S_i\in \mb{K}^S$ for any $k\in[0,\delta k]$ since $-\alpha_i\le -\alpha_{ref}$. Therefore $\sum\limits_{k\in [0, \delta k]}-\alpha_{\pi_k}\mu_{\pi_k} V_{\pi_0}\le \sum\limits_{k\in [0, \delta k]}-\alpha_{ref, \pi_0} \mu_{\pi_k} V_{\pi_0}$ for any $\mu_{\pi_k} > 0$ since $-\alpha_i\le -\alpha_{ref}$. This ensures the desired timely convergence while switching.
\end{proof}
\noindent Note that the SCAP, as defined in Eq.~\ref{eq:scap}, {\em greedily} chooses the next BSC to impose maximum decay at every iteration, which {\em may not} ensure the fastest recovery. But as shown in Lemma~\ref{lemSwitch}, it ensures both safety and recovery within the desired deadline. 
In the next section, we present the algorithmic framework for SCAP to constrain switching to maintain resource awareness.
\subsection{Step 4: Incorporating Resource-awareness in Algorithmic Framework for SCAP}
\label{sec:Algo}
As discussed in the last section, the state-dependent controller activation policy (SCAP) decides which controller to switch to based on the current state. 
The decision is implied by the distance of the current trajectory from the boundaries of $\R{P}$ or $\R{S}$. 
Essentially, SBFs help us measure this distance; the more negative its values, the safer the states. However, in our case, the SBFs are specific to BSCs and denote the distance of the current trajectory from the largest SFIR boundary of the currently active BSC, rather than $\R{S}$.
Therefore, judging safety based on their values throughout the entire SOR is not a suitable approach. To this end, we design a {\em global logarithmic barrier function (GLBF)} to measure this notion of safety in runtime. It is defined as,
   $B_g(x)=\sum\limits_{i\le n} -log(a_S^{(i)} x_i - b_S^{(i)})$
%
where $a_S^{(i)}, b_S^{(i)}$ are as defined in Eq.~\ref{eq:LMIsafe2}. 
Note that it satisfies the three required properties for SBFs as given in Eq.~\ref{eq:BF}. Hence, the lower the GLBF's value for an input state, the safer the state is. We can reduce resource usage by switching to an SBF with a lower rate, as the system recovers to a safer state.
This resource-aware decision is invoked in the SCAP before the BSCs exceed the budgeted bandwidth amount for the current control task. 
This is demonstrated in Fig.~\ref{fig:scap}. At $time = 0$, the control tasks (marked with black horizontally-stiped boxes) corresponding to BSC $K^S_2$ are executed once in every $h_2$ sampling period. As the utilisation budget exceeds and the trajectory progresses, task instances corresponding to another BSC $K^S_1$ with a higher sampling period $h_1 (>h_2)$ start executing. The switch is marked with a blue star in Fig.~\ref{fig:scap}. This is decided based on the decay measured using $B_g(.)$. 
Therefore, in addition to the earlier condition of activating BSCs with the highest decay, as presented in Eq.~\ref{eq:scap}, we also introduce a utilisation minimising criterion based on the GLBF decay values. This switching naturally frees processing bandwidth due to fewer control task executions. As the trajectory enters POR $\R{P}$ 
within the recovery deadline of $0.5s$, controller tasks corresponding to POC start executing (marked with a black star). We present the algorithmic framework that implements SCAP in Algo.~\ref{algoscap}. 
\par\noindent $\blacksquare$ {\bf Algorithmic Framework:} Algo.~\ref{algoscap} takes the following inputs:  {\bf (1)} The set of BSCs $\mb{K}^S$, {\bf (2)} set of its corresponding QLFs $\mb{V}$, {\bf (3)} set of their respective SBFs $\mb{B}$, {\bf (4)} set of periodicities $\mb{h}$, {\bf (5)} set of respective QLF decays $\alpha vals$, {\bf (6)} Recovery deadline $\delta t$, {\bf (7)} The POC, {\bf (8)} a potentially time-varying utilisation budget for the current closed loop $U_b$, {\bf (9)} the GLBF $B_g(.)$ for the system, {\bf (10)} the POR $\R{P}\subset \R{S}$. The algorithm activates a suitable controller (POC or BSCs) at every $k$-th sampling iteration. We start by initialising an array $lb$ of length $\delta t / \max\limits_{h_i\in\mb{h}}(\mb{h})$, with 0 at every index (see line~\ref{algscapInit}).  
\begin{algorithm}[!b]
\footnotesize
\begin{algorithmic}[1]
\Require{The sets of: BSCs $\mb{K}^S$, their corresponding QLFs $\mb{V}$, SBFs $\mb{B}$, periodicities $\mb{h}$, QLF decays $\mb{\alpha vals}$, and Recovery deadline $\delta t$, The POC $K$, utilisation budget for the current control loop $U_b$, the GLBF $B_g(.)$, POR $\R{P}$}
\Ensure{Activate Suitable Controller (BSCs/POC) at every Sampling Instance}
\State $lb\gets \delta t / \max\limits_{h_i\in\mb{h}}(\mb{h})$ length empty array, $\Delta_{lb}\gets 0$, $d_{z}\gets 2,\ decay \gets 0,\ idx \gets 0$
\label{algscapInit}
\State $\cl{K}^S\gets \{\tupleangle{h_i, K^S_i, V_i, \alpha_i, B_i}\forall h_i\in \mb{h}\}$ from respective sets, sorted in increasing order of $h_i$\label{algscapInit2}
    \For{each $k$-th sampling iteration}\label{algscapForIterStart}
         \State $lb\gets$ shift each  element to left and add, $B_g(x[k])$ at end, $\Delta_{lb}\gets \min\limits_{i< \abs{lb}}\frac{lb[i]-lb[i-1]}{(\abs{lb}-1)}$\label{algscapGlbfComp}
         \State $util \gets$ Compute processor utilisation for current BSC task, $d_z\gets d_z-1$ \label{algscapFindUtil} 
         \If{$ x[k]\in\R{P}$ \& $d_{z}\le 0$ }\label{algscapIfInPor}
             Activate POC \Comment{when inside POR}
        \Else\If{$d_{z}=0$ or $util> U_b$}\label{algscapIfNotInPor}\Comment{when outside POR use SCAP from Eq.~\ref{eq:scap}}
             \For{each tuple in $\cl{K}^S$}\label{algscapForeachSbc}\Comment{search BSC with highest decay, safe in current region}
                \If {$B_i(x[k])\leq 0$ \& $util\le U_b$}\label{algscapIfsafeSbcLessutil}\Comment{search among BSCs for current region}
                    \If{$decay < \alpha_iV_i(x)$}\label{algscapIfmostLFdecay}
                    \State $decay \gets \alpha_iV_i(x),\ idx \gets i, d_{z}=2$\label{algscapIfmostLFdecaythen}
                    \EndIf\label{algscapIfmostLFdecayEnd}
                \Else\If {$B_i(x[k])\leq 0$ \& $util> U_b$}\label{algscapIfsafeSbcMoresutil}\Comment{if utilisation beyond budget}
                    \If{($\Delta_{lb}<0$ $\Rightarrow h_i > h_{idx}$) \& $(decay < \alpha_iV_i(x))$} \label{algscapIfmostLFdecayHighper} 
                        \State $decay \gets \alpha_iV_i(x),\ idx \gets i, d_{z}=2$\label{algscapAssignMostLFdecay}\Comment{select BSCs with lower rates}
                    \EndIf\label{algscapIfmostLFdecayHighperEnd}
                    \If{$\Delta_{lb}\ge 0$} Notify to Reduce Utilisation of less critical programs\label{algscapIfPositiveDecay}\EndIf
                \EndIf\EndIf\label{algscapIfsafeSbcMoresutilEnd}
             \State {\bf end for \bf} \EndFor \label{algscapForeachSbcEnd}
             \State activate $K^S_{idx}\in \mb{K}^S$\label{algscapActvtSbc}
         \EndIf\EndIf\label{algscapIfNotInPorEnd}
    \EndFor 
    \label{algscapForIterEnd}
\end{algorithmic}
\caption{Resource-Aware Safe Controller Activation in Runtime}
\label{algoscap}
\end{algorithm}
The length is determined to capture the value of GLBF for the minimum number of sampling iterations required for any BSC to recover within the deadline $\delta t$, which is the number of iterations required for the BSC with the longest sampling period. We initialise $\Delta_{lb}$ to store the average decay of GLBF during recovery. A variable $d_z$ is defined, which holds the minimum number of sampling iterations that an BSC should remain activated. This is used to avoid the {\em Zeno phenomena} that can cause an infinite number of switches between BSCs within a finite time. The minimum number of sampling iterations $d_z$ is initialised with a constant value of 2, but it can also be parameterised as an input. The variables $decay$ and $idx$ are initialised to a value of 0. The $decay$ variable stores the QLF decays imposed by different BSCs at the current state, while the $idx$ variable keeps track of the index of the BSC with the highest decay from the list of BSCs $\mb{K}^S$. In line~\ref{algscapInit2}, we create a tuple for each BSC gain from the set $\mb{K}_S$. Each tuple contains the following information: the BSC gain, its periodicity from $\mb{h}$, its QLF from $\mb{V}$, its QLF decay from $\mb{\alpha vals}$, and its SBF from $\mb{B}$. We then store these tuples in a list $\cl{K}^S$, arranging them in increasing order of their periodicities $h_i \in \mb{h}$.
\par For each sampling iteration (from lines~\ref{algscapForIterStart}-\ref{algscapForIterEnd}), first we compute the GLBF value for the current system state and append it to the fixed-length array $lb$, by removing the first element and shifting each of the elements to their left (see line~\ref{algscapGlbfComp}). This is done to compute the average decay over the $\delta t$ duration, i.e., by dividing the difference between consecutive elements of $lb$ and dividing it by $\abs{lb}-1$, where $\abs{lb}$ signifies the length of the fixed-length $lb$ array. In the next line, the processor utilisation of the current task is computed and stored in $util$ (line~\ref{algscapFindUtil}). In the same line, the value of the minimum required number of sampling iterations $d_z$ is also decremented by 1 as we run any controller for one sampling iteration. If the current state is in $\R{P}$ and $d_z$ has a value less than or equal to 0, indicating that we can switch to a different controller, we activate the POC. When the current state is outside $\R{P}$, i.e., in $\Runsf$, we need to activate the BSCs (lines~\ref{algscapIfNotInPor}-\ref{algscapIfNotInPorEnd}). This is done by checking each BSC tuple from $\cl{K}^S$ following the SCAP definition as derived in Eq.~\ref{eq:scap} (see lines~\ref{algscapForeachSbc}-\ref{algscapForeachSbcEnd}). 
\begin{figure*}[!ht]
    \centering
    \begin{subfigure}[b]{0.49\linewidth}
        \includegraphics[clip,width=\linewidth]{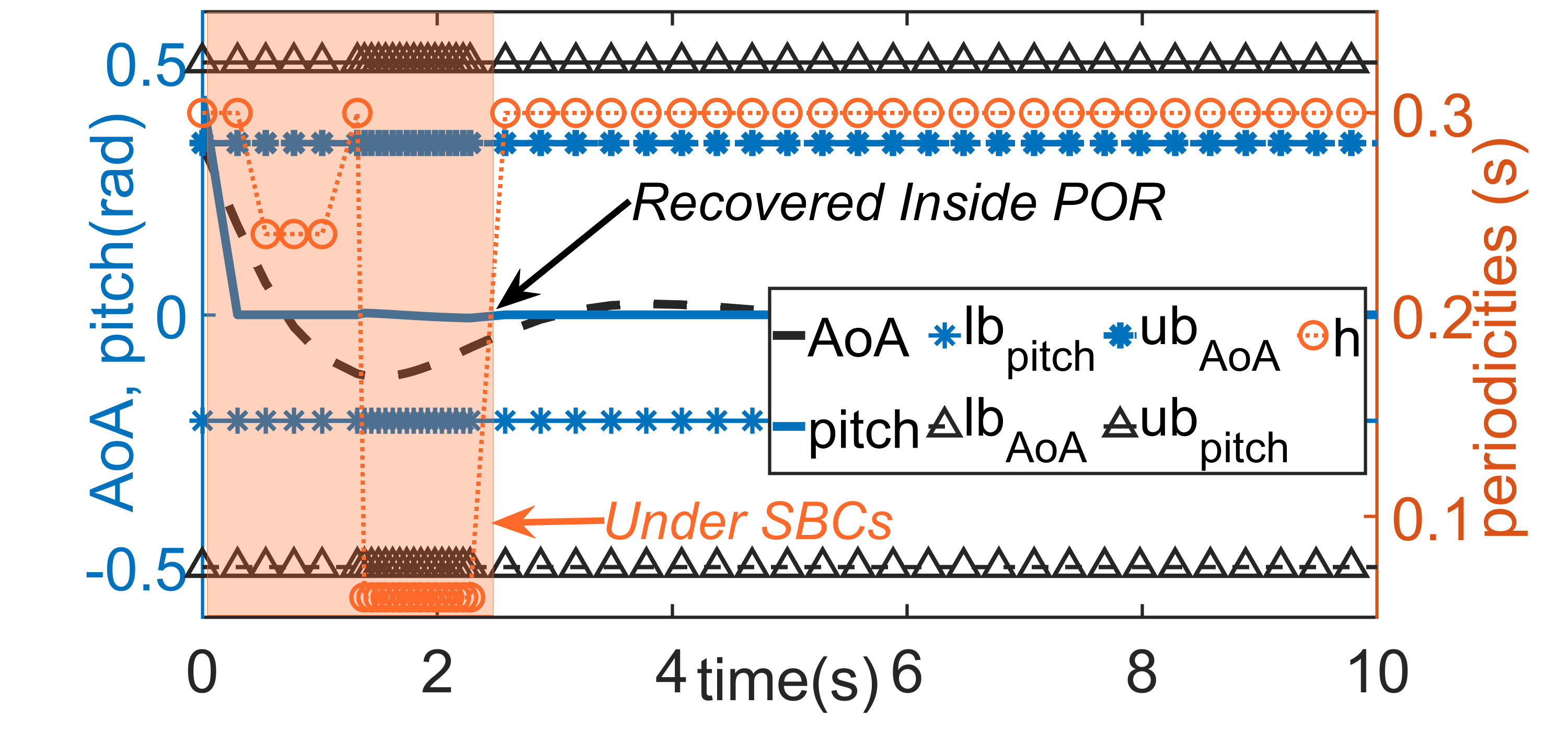}
    \label{fig:alcsbc2}
    \end{subfigure}
    \hfill
    \begin{subfigure}[b]{0.49\linewidth}
        \includegraphics[clip,width=\linewidth]{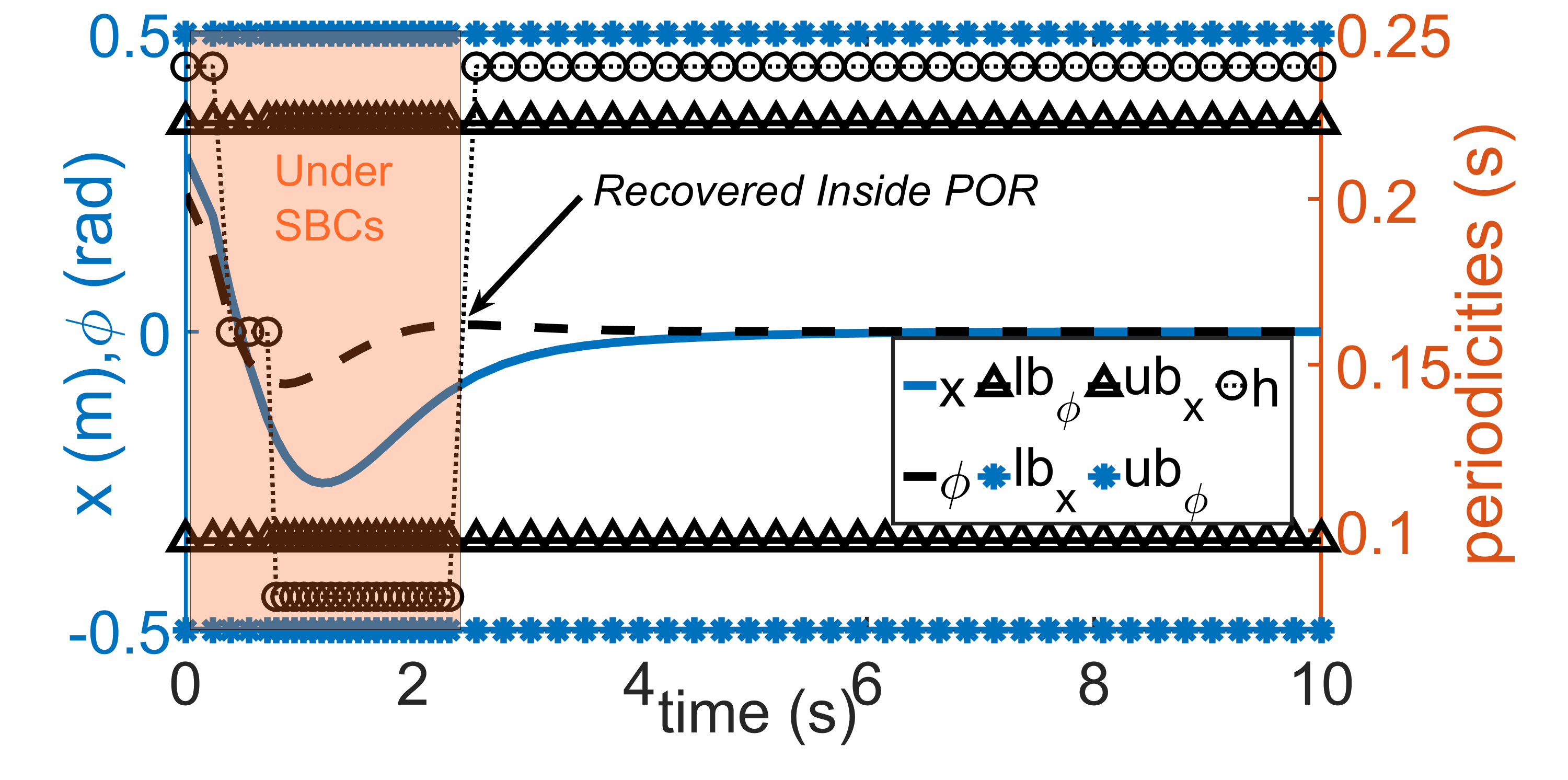}
    \label{fig:ipsbc2}
    \end{subfigure}
    \caption{ Adaptive Scheduling of BSCs:$\quad$ (a) Aircraft Longitudinal Dynamics (on left),$\quad$ (b) Inverted Pendulum Dynamics (on right)}
    \label{fig:recoveryadaptive}
\end{figure*}
\par Now, there are two possibilities: {\em 1.} when the processor utilisation for this feedback control loop $util$ is within the budget $U_b$, which is implemented in lines~\ref{algscapIfsafeSbcLessutil}-\ref{algscapIfmostLFdecayEnd} and {\em 2.} when $util$ is beyond the budget $U_b$, which is implemented in~\ref{algscapIfsafeSbcMoresutil}-\ref{algscapIfsafeSbcMoresutilEnd}. For {\em case 1.}, if the current state is in the SFIR of an BSC, we measure decay that will be imposed by it using the variable $decay$ at the current state. The index of the BSC with the highest decay from $\cl{K}^S$ is greedily stored in the variable $idx$ (see line~\ref{algscapIfmostLFdecaythen}). For {\em case 2.}, along with checking whether the current region is in the SFIR of an BSC, we also check whether the average decay in the last $\delta t$ duration as stored in $\Delta_{lb}$ is negative. If $\Delta_{lb}$ is negative (meaning safe), then we only consider controllers with higher sampling periods or lower execution rates compared to the currently activated BSC and compare the decays imposed by them at the current sampling iteration (see line~\ref{algscapIfmostLFdecayHighper}). The index of the BSC with the highest decay and higher periodicity is then stored in the variable $idx$ (see line~\ref{algscapAssignMostLFdecay}). The if condition in line~\ref{algscapIfmostLFdecayHighper} also takes care of the case when $\Delta_{lb}$ is positive. In such cases, we select the BSC with the highest decay, regardless of its periodicity, and notify the scheduler to decrease the utilisation budget of other less critical tasks (see line~\ref{algscapIfPositiveDecay}). When selecting a new controller, we always set the variable $d_z$ to 2 to prevent the controller from switching back and forth frequently. This value can be modified based on the system's behaviour. 
Finally, the selected BSC is activated for the following iteration (line~\ref{algscapActvtSbc} in Algo.~\ref{algoscap}).
\par\noindent $\blacksquare$ \textbf{Correctness, Completeness and Complexity Analysis:} Algo~\ref{algoscap}  implements SCAP to switch between BSCs (and POC) in both cases, i.e., while the utilisation is within $U_b$ and while it is beyond budget. The correctness of the algorithm is ascertained by the fact that  Lemma~\ref{lemSwitch} provides the stability and safety guarantee of SCAP. 
It may be observed that if a schedulable and valid controller switching solution exists for a given problem instance, the algorithm will report it. 
In the absence of such solutions, the algorithm stops searching and suggests a change in design parameters (line~\ref{algscapIfPositiveDecay}). 
%
\par The algorithm does the following in every sampling iteration: {\bf (i)} evaluates the average decay by storing GLBF for a sliding window of $\scriptstyle{\delta t/\max\limits_{h_i}h_i}$ sampling iterations, {\bf (ii)} evaluates whether the system state in the current iteration is inside $\R{P}$, {\bf (iii)} evaluates and compares the decay that will be imposed by BSC with different sampling periods. For part {\bf (i)}, we can compute it in $O(\scriptstyle{\frac{\delta t}{\max\limits_{h_i\in \mb{h}}h_i}})$. For part {\bf (ii)}, if there exists a closed-form polytopic expression of $\R{P}$ with $q>0$ inequalities, the worst-case complexity of the point inclusion check is polynomial in the number of inequalities, i.e., $O(q)$. Lastly, for part {\bf (iii)}, in the worst case, the algorithm evaluates the decay values for each BSC in every sampling iteration and compares them. The complexity therefore becomes $O(\abs{\mb{K}^S})$, where $\abs{\mb{K}^S}$ denotes its length of $\mb{K}^S$. Therefore, the overall worst-case complexity is $O(\scriptstyle{\frac{\delta t}{\max\limits_{h_i\in \mb{h}}}h_i}+\abs{\mb{K}^S}+q)$. 
In the following section, we evaluate the synthesised BSCs and Algo.~\ref{algoscap} by simulating them for two case studies.
\begin{figure*}[!ht]
    \centering
        \begin{subfigure}[b]{0.48\linewidth}
        \includegraphics[clip,width=\linewidth]{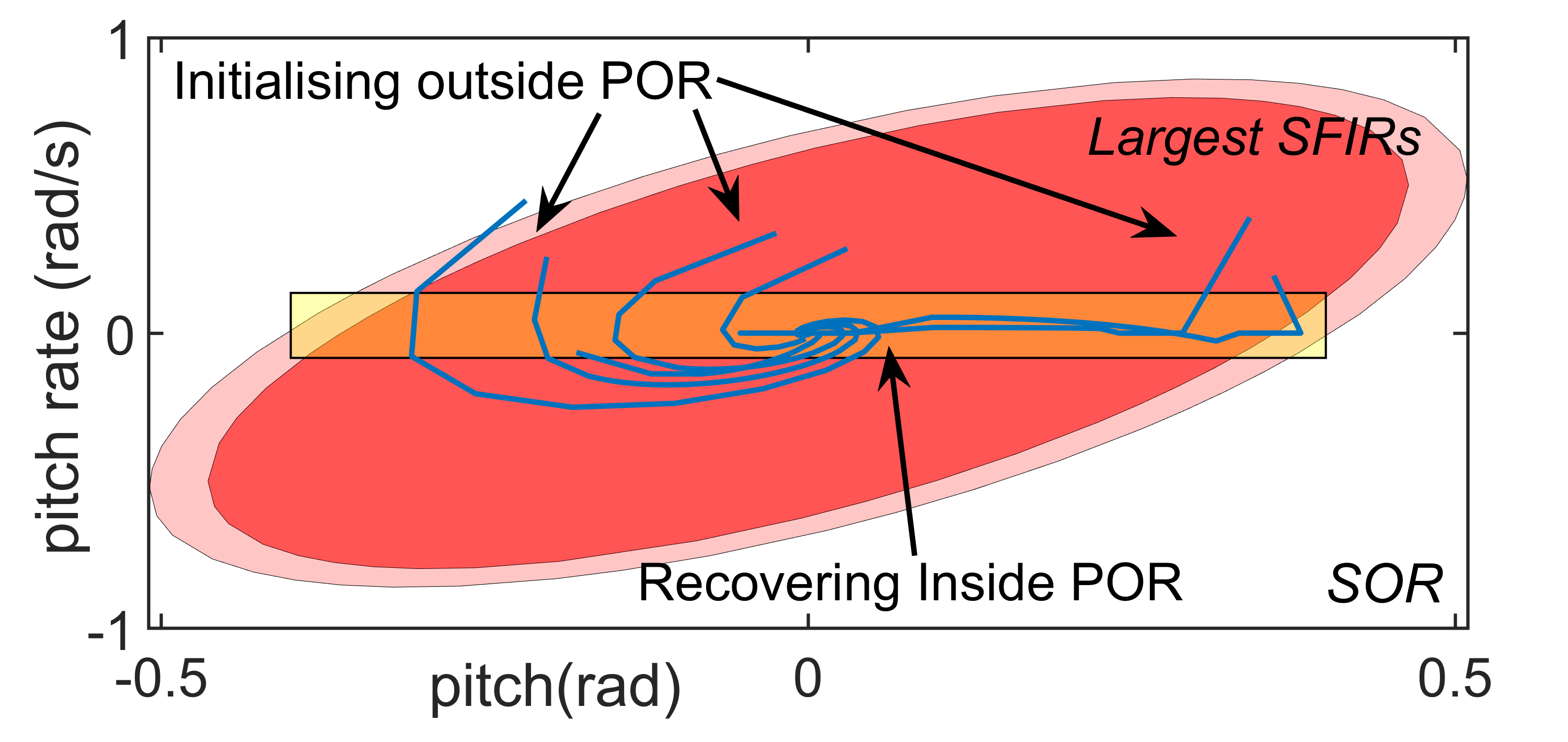}
        \label{fig:alcsbc1}
    \end{subfigure}
    \hfill
    \begin{subfigure}[b]{0.48\linewidth}
        \includegraphics[clip,width=\linewidth]{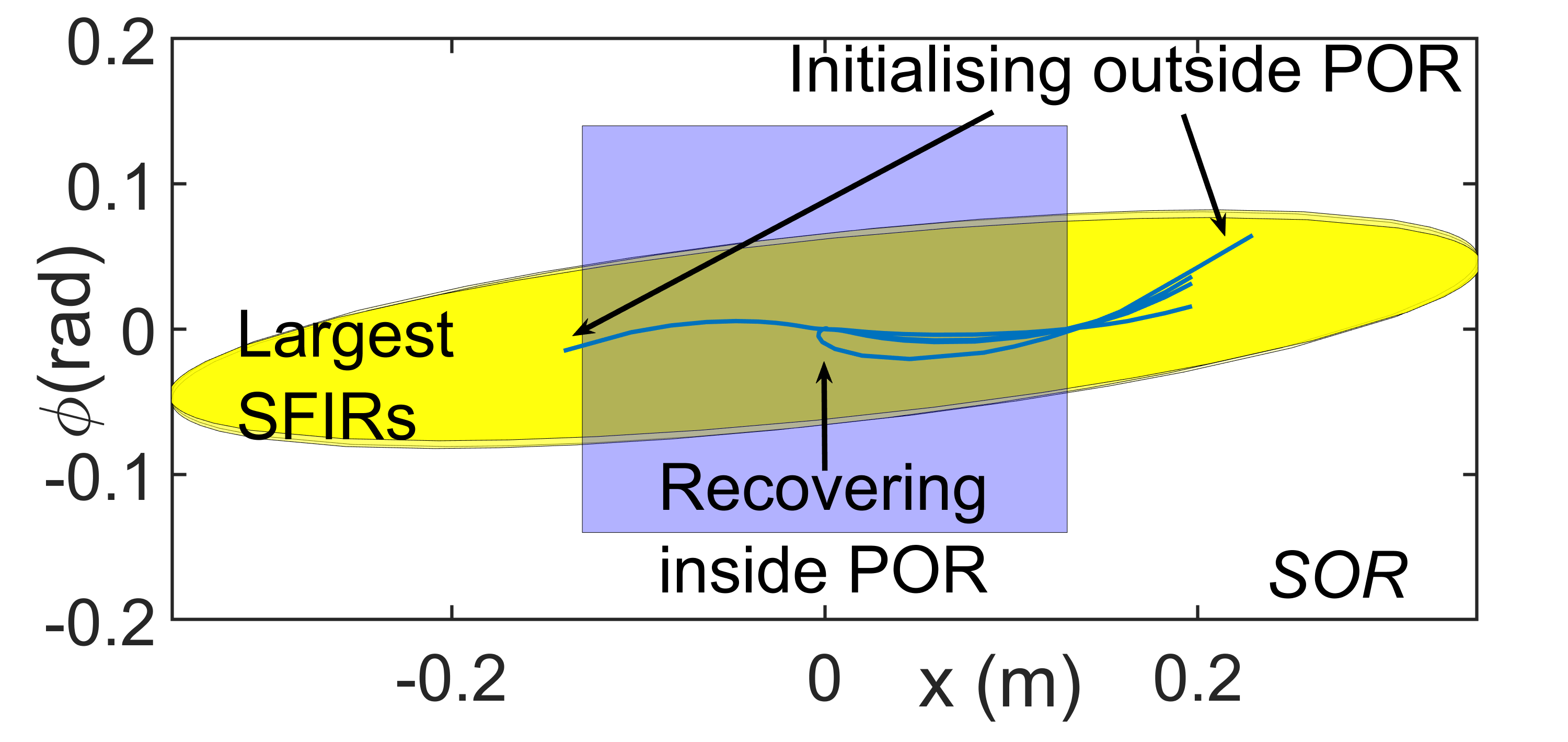}
        \label{fig:ipsbc1}
    \end{subfigure}
    \caption{State Recovery Under BSCs:$\quad$ (a) Aircraft Longitudinal Dynamics (on left),$\quad$ (b) Inverted Pendulum Dynamics (on right)}
    \label{fig:recoverysatates}
\end{figure*}
\vspace{-5mm}
\section{Experiments}
\label{secExp}
We synthesise and implement BSCs and their activation framework for aircraft longitudinal dynamics (ALD)~\cite{lavretsky2024robust} and inverted pendulum dynamics (IPD) ~\cite{ctms}, linearised around specific equilibrium points. The BSC synthesis Prob.~\ref{prob:synth} is modelled using YALMIP~\cite{Lofberg2004} and is solved using Mosek in MATLAB. For evaluation, Algorithm~\ref{algoscap} is subsequently simulated using MATLAB with BSCs synthesised for both systems. 
\par\noindent $\blacksquare$ {\bf Aircraft Longitudinal Dynamics Control:} The control system for ALD aims to achieve a desired pitch angle (in rad) by actuating proper elevation (in m) input. The system is linearised around a pitch angle of $0.5$ rad. We derive the controllers for this system with recovery deadlines of 1 second and 2.5 seconds, given a base periodicity of $20$s (see Sec.~\ref{sec:switchctrl}). The chosen periodicities to synthesise BSCs are the multiples of the base periodicity and inside the range of $[20- 300]$ms. The safety boundaries for the four measurable states of ALD are: $[-100, 100]$ kmph for airspeed, $[-0.12,0.36]$ rad for angle of attack (AoA), $[-1,1]$ rad/s for pitch rate, and $[-0.5,0.5]$ rad for pitch angle. Note that our synthesis framework was able to synthesise controllers having $\{20, 40, 60, 80, 120\}$ ms sampling periods that safely meet the required recovery deadlines of $1$s and $\{20, 40, 60, 80, 120, 160, 200, 240, 300\}$ms sampling periods that safely meet the required recovery deadlines of $2.5$s. As discussed, these controllers have their respective SFIRs in the state space and need to be applied by the switching policy (Algorithm~\ref{algoscap}) based on the utilisation bound as well as the current state. 
%
\par In the left-side figure in Fig.~\ref{fig:recoveryadaptive}, we plot the evolution of AoA (dashed black line) and pitch angle deviation from the desired equilibrium $0.5$ rad (blue line) under the operation of BSCs (in the left y-axis) for a recovery deadline of $2.5$s w.r.t. time (in s). The changing periodicities of BSCs (brown plots with circle markers) are plotted against time (x-axis) on the right y-axis. The upper and lower safety boundaries for AoA are plotted in blue lines with star-shaped markers, and the safety bounds for pitch angle are plotted in black with triangular markers. The periodicities $h_i$ are plotted in brown, with circle markers. The pitch angle POR and AoA POR are $[-0.35,0.35]$ rad and $[-0.084,1.36]$ rad, respectively. The SCAP implementation activates BSCs when the system is outside this POR (the light brown area in Fig.~\ref{fig:recoveryadaptive}).
%
\par  Potentially unsafe trajectories are simulated by initialising system trajectories from $\Runsf$, i.e., inside SOR but not in POR. In this simulation, the time-varying utilisation budget $U_b$ considered allows for a minimum periodicity of 60 ms in $t\in [0,1.32]$s, 300 ms in $t\in [1.33,1.87]$s and 60 ms in $t\in [1.88,2.5]$s. With the system $\notin$ POR at the start of the simulation, the SCAP implementation greedily activates the BSC with 240 ms, as it imposes the highest decay rate among the BSCs having current state within their SFIRs. 
This does not restore the system states fully inside POR.
To satisfy the input utilisation budget at the next time instant, SCAP activates an BSC with a higher periodicity of 300 ms in the following iteration since the system states are relatively {\em safer} than earlier. 
At this point, to utilise the higher available bandwidth in the next sampling iteration, SCAP again switches to 60 ms BSC (with higher frequency) and uses it for the next few iterations. As the trajectory enters POR within $2.5$s as designed, the SCAP reactivates the POC with $300$ms sampling period. 
\par To demonstrate the overall efficacy of the proposed synthesis framework, in the left-side figure in Fig.~\ref{fig:recoverysatates}, we plot the recovery of multiple potentially unsafe trajectories in blue. Fig.~\ref{fig:recoverysatates} plots pitch angle deviation on the x-axis and pitch rate (for ease of visualisation, as the AoA SOR is relatively narrow) on the y-axis. As the trajectories recover by activating different BSCs, their corresponding largest ellipsoidal SFIRs are also plotted in red. The yellow box denotes POR. Notice that every trajectory starting outside the POR is able to recover inside the POR by activating proper BSCs. 
\par\noindent $\blacksquare$ {\bf Inverted Pendulum Control:} The controller for IPD applies force to maintain its position at $x=0$. The dynamics are linearised around the equilibrium point of phase angle $\phi=\pi$ rad. The BSCs for this system are derived with recovery deadlines of $2s$ and $2.5s$. 
The safety boundaries for the four measurable states of IPD are: $[-0.5,0.5]$ m for position $x$, $[-0.5,0.5]$ m/s for speed, $[-0.35,0.35]$ rad for phase angle $\phi$, and $[-0.35,0.35]$ rad/s for angular velocity. Given a base periodicity of $20$ ms and a range of $[20,300]$ms, we solve Prob.~\ref{prob:synth} to synthesise multiple BSCs. Our synthesis framework was able to synthesise BSCs with \{20, 40, 60\} ms sampling periods that safely meet the required recovery deadline of $2$s and BSCs with $\{20, 40, 60, 100, 120, 160, 200, 240\}$ ms sampling periods that meet the deadline of $2.5$s. For the recovery deadline of $1s$ synthesised controllers, they are unusable due to very small SFIRs. 
\par In the left y-axis of the right-side figure in Fig.~\ref{fig:recoveryadaptive}, the position $x$ (in solid blue) and phase angle $\phi$ (in dashed black) under the operation of synthesised BSCs are plotted, whereas the periodicities of active BSCs (black plot with circle markers) are plotted on the right y-axis. In this simulation, $U_b$ is assumed constant and allows a minimum period of $40ms$. Notice that to handle the deviation beyond POR ($[-0.14,0.14]$ for both $x$, $\phi$), first an BSC with periodicity 160 ms and then an BSC with periodicity 80 ms are activated. Note that, although the BSC with 80 ms periodicity imposes a higher decay compared to the 160 ms BSC, initially SCAP could not activate the BSC with 80 ms periodicity, since the system states were not in its SFIR (but in the SFIR of the 160 ms BSC). This ensures safe state recovery within a 2.5-second mark, as guaranteed by the synthesis method. The figure in the right half of Fig.~\ref{fig:recoverysatates} plots multiple trajectories of IPD states under the operation of synthesised BSCs. We plot $x$ on the x-axis and $\phi$ on the y-axis. The overlapping yellow ellipses show the largest SFIRs of the active BSCs at different sampling periods. The purple rectangle denotes the POR of IPD. As can be seen, all trajectories that lie outside the POR are (timely) recovered inside the POR without compromising safety.
\par \noindent $\blacksquare$  \textbf{Limitations and Advantages:} 
The state-of-the-art works~\cite{dai2024verification,garg2021robust,sudvarg2025integrated} primarily focus on SOS-based compatible Lyapunov and barrier synthesis methods or formulate quadratic program (QP)-based safe controller synthesis methods. Some~\cite{garg2021robust}  enforce timely recovery with predefined LFs and BFs, promising similar results to ours. These methods are also applicable to nonlinear dynamics, unlike the proposed solution, which is applicable only to linearised hybrid dynamics, since a candidate LF for such systems is relatively easy to obtain.
%
\par Our method relies on SDP-based synthesis. Therefore, compared to the state-of-the-art (QP, SOS), it requires less time to solve (takes $\sim 0.43$s in an Intel Core i5 processor with 8 GB of RAM, QP takes $\sim 0.85$s). Additionally, the proposed synthesis framework repurposes QLFs as SBFs for various sampling rates, thereby increasing the feasibility of a larger number of schedulable safe control solutions that individually and together cover a broader region within the safe state space, ensuring timely recovery from those regions.

\section{Concluding Remarks}
\label{secConcl}
We present a novel {\em backup safe controller} synthesis technique for simplex architectures that focuses on safe and timely recovery of linearised time-invariant hybrid system trajectories utilising quadratic Lyapunov function-based safety barriers. To obtain more feasible and schedulable options in resource-limited platforms, we solve the synthesis problem for different control data computation rates and devise a {\em state-dependent controller activation policy} to enable a greedy yet safe switching between the synthesised controllers in the presence of a strict utilisation budget. In the future, we intend to devise a similar safe control solution using higher-order Lyapunov and barrier functions for nonlinear hybrid dynamics and validate it through real-world demonstrations.

\newpage
\bibliographystyle{ACM-Reference-Format}
\bibliography{references}

\end{document}